\documentclass[fleqn,11pt]{article}%
\usepackage{amsmath}
\usepackage{float}
\usepackage{theorem}
\usepackage{graphicx}%
\usepackage{amsfonts}%
\usepackage{amssymb}
\theorembodyfont{\upshape}
\newtheorem{theorem}{Theorem}

\newtheorem{assumption}[theorem]{Assumption}

\newtheorem{corollary}[theorem]{Corollary}

\newtheorem{definition}[theorem]{Definition}
\newtheorem{example}[theorem]{Example}

\newtheorem{notation}[theorem]{Notation}

\newtheorem{remark}[theorem]{Remark}

\newenvironment{proof}[1][Proof]{\textbf{#1.} }{\ \rule{0.5em}{0.5em}}
\advance\topmargin by -0.75truein
\advance\textheight by 1.25truein
\begin{document}

\begin{center}
{\LARGE On Limits to the Scope of the \textit{Extended Formulations}
``Barriers''}{\Large \medskip\medskip}

Moustapha Diaby

OPIM Department; University of Connecticut; Storrs, CT 06268\\[0pt]moustapha.diaby@business.uconn.edu\bigskip

Mark H. Karwan

Department of Industrial and Systems Engineering; SUNY at Buffalo; Amherst, NY
14260\\[0pt]mkarwan@buffalo.edu\bigskip\ \ 
\end{center}

\textsl{Abstract}{\small : In this paper, we introduce the notion of
\textit{augmentation} for polytopes and use it to show the error in two
presumptions that have been key in arriving at over-reaching/over-scoped
claims of ``impossibility'' in recent \textit{extended formulations} (EF)
developments. One of these presumption is that: ``If Polytopes }${\small P}%
${\small  and }${\small Q}${\small  are described in the spaces of variables
}${\small x}${\small  and }${\small y}${\small  respectively, and there exists
a linear map }${\small x=Ay}${\small  between the feasible sets of
}${\small P}${\small  and }${\small Q}${\small , then }${\small Q}${\small  is
an EF of }${\small P}${\small ''. The other is: ``(An \textit{augmentation} of
Polytope }${\small A}$ {\small projects to Polytope }${\small B}${\small ) ==%
$>$
(The external descriptions of }${\small A}$ {\small and }${\small B}$
{\small are related)''. We provide counter-examples to these presumptions, and
show that in general: (1) If polytopes can \textit{always} be arbitrarily
\textit{augmented} for the purpose of establishing EF relations, then the
notion of EF becomes degenerate/meaningless in some cases, and that: (2) The
statement: ``(Polytope }${\small B}$ {\small is the projection of an
\textit{augmentation} of Polytope }${\small A}${\small ) ==%
$>$
(Polytope }${\small B}$ {\small is the projection of Polytope }${\small A}%
${\small )'' is not true in general (although, as we show, the converse
statement, ``(}${\small B}${\small  is the projection of }${\small A}%
${\small ) ==%
$>$
(}${\small B}${\small  is the projection of every augmentation of }%
${\small A}${\small )'', is true in general). We illustrate some of the ideas
using the \textit{minimum spanning tree problem}, as well as the ``lower
bounds'' developments in Fiorini\textit{\ et al.} (2011; 2012), in
particular.\bigskip}

\textsl{Keywords:}\textbf{\ }{\small Linear Programming; Combinatorial
Optimization; Traveling Salesman Problem; TSP; Computational Complexity,
Extended Formulations.}$\medskip$

\section{Introduction\label{Introduction_Section}}

There has been a renewed interest in Extended Formulations (EF's) over the
past 3 years (see Conforti \textit{et al}. (2010), Vanderbeck and Wolsey
(2010), Fiorini \textit{et al.} (2011; 2012), and Kaibel (2011), for example).
Despite the great importance of the EF paradigm in the analysis of linear
programming (LP) and integer programming (IP) models of combinatorial
optimization problems (COP's), the clear definition of its scope of
applicability has been largely an overlooked issue. The purpose of this paper
is to make a contribution towards addressing this issue. Specifically, we will
show that the notion of an EF can become \textit{ill-defined} and degenerate
(and thereby lose its meaningfulness) when it is being used to relate
polytopes involved in alternate abstractions of a given optimization problem.
Because most of the papers on EF's focus on the TSP specifically, we will
center our discussion on the TSP. However, the substance of the paper is
applicable for other NP-Complete problems.

It should be noted that, in this paper, we are not concerned with and do not
claim the correctness or incorrectness of any particular model that may have
been developed in trying to address the ``$P=NP$'' question. Our aim is,
strictly, to bring attention to limits to the scope within which EF Theory is
applicable when attempting to derive bounds on the size of the description of
a polytope.

The plan of the paper is as follows. First, in section
\ref{Background_Section}, we will review the basic definitions and notation
and show, in particular, the error of any notion there may be, suggesting an
``impossibility'' of abstracting the TSP optimization problem over a polytope
of polynomial size, by showing that TSP tours can be represented (by
inference) independently of the (standard) TSP polytope. Then, we will
introduce the notion of ``polyhedron augmentation'' in section
\ref{Polytope_Augmentation_Subsection}, and use it (in section
\ref{Ill_Dfn_Condition_Subsection}) to develop our results on the condition
about EF's becoming \textit{ill-defined}. Finally, we will offer some
concluding remarks in section \ref{Conclusions_Section}.\medskip

The general notation we will use is as follows.

\begin{notation}
\label{TSP_Gen'l_Notations}\ \ 
\end{notation}

\begin{enumerate}
\item $\mathbb{R}:$ \ Set of real numbers; \ 

\item $\mathbb{R}_{\mathbb{\nless}}:$ \ Set of non-negative real numbers;

\item $\mathbb{N}:$ \ Set of natural number;

\item $\mathbb{N}_{\mathbb{+}}:$ \ Set of positive natural numbers;

\item $``\mathbf{0}":$ \ Column vector that has every entry equal to $0$;

\item $``\mathbf{1}":$ \ Column vector that has every entry equal to $1$;

\item $Conv(\cdot):$ \ Convex hull of ($\cdot$).\medskip
\end{enumerate}

\section{Background overview\label{Background_Section}}

\subsection{Basic definitions\label{Basic_Dfns_SubSection}}

\begin{definition}
[``traditional $x$-variables'']\label{traditional x-variables}We will
generically refer to $2$-indexed variables that have been used in traditional
IP formulations of the TSP to represent inter-city travels as ``traditional
$x$-variables.'' In other words, $\forall(i,j)\in\{1,...,n\}^{2}:$ $i\neq j$,
we will refer to the $0/1$ decision variables that are such that the
$(i,j)^{th}$ entry of their vector is equal to ``$1$'' iff there is travel
from city $i$ to city $j$, as the ``traditional $x$-variables,'' irrespective
of the symbol used to denote them.
\end{definition}

\begin{definition}
[``Standard TSP Polytope'']\label{Standard_TSP_Polytope} Let $\mathcal{A}%
:=\left\{  (i,j)\in\Omega^{2}:i\neq j\right\}  $ denote the set of arcs of the
complete digraph on $\Omega.$ Denote the characteristic vector associated with
any $F\subseteq$ $\mathcal{A},$ by $x^{F}$ (i.e., $x_{ij}^{F}\in\{0,1\}$ is
equal to $1$\ iff $(i,j)\in F$). Assume (w.l.o.g.) that the TSP tours (defined
in terms of the arcs) have been ordered, and let $\mathcal{T}_{k}\subset$
$\mathcal{A}$ denote the $k^{th}$ tour. The ``Standard TSP Polytope'' (in the
asymmetric case) is defined as $Conv\left(  \left\{  x^{\mathcal{T}_{k}}%
\in\mathbb{R}^{n(n-1)},\text{ }k=1,\ldots,n!\right\}  \right)  $\ (see Lawler
\textit{et al}. (1985, pp. 257-258), and Yannakakis (1991, p. 441), among others).
\end{definition}

\begin{definition}
[``Standard EF Definition'' (Yannakakis\ (1991); Conforti \textit{et al}.
(2010; 2013))]\label{Extended_Polytope_Dfn}An ``extended formulation'' for a
polytope $X$ $\subseteq$ $\mathbb{R}^{p}$ is a polyhedron $U$ $=$ $\{(x,w)$
$\in$ $\mathbb{R}^{p+q}$ $:$ $Gx$ $+$ $Hw$ $\leq$ $g\}$ the projection,
$\varphi_{x}(U)$ $:=$ $\{x\in\mathbb{R}^{p}:$ $(\exists w\in\mathbb{R}^{q}:$
$(x,w) $ $\in$ $U)\},$ of which onto $x$-space is equal to $X$ (where $G$
$\in\mathbb{R}^{m\times p},$ $H\in\mathbb{R}^{m\times q},$ and $g\in
\mathbb{R}^{m}$).
\end{definition}

\begin{definition}
[``Alternate EF Definition \#1\textit{''} (Kaibel (2011); Fiorini \textit{et
al}. (2011; 2012))]\label{Extended_Polytope_Dfn2}A polyhedron $U$ $=$
$\{(x,w)$ $\in$ $\mathbb{R}^{p+q}$ $:$ $Gx$ $+$ $Hw$ $\leq$ $g\}$ is an
``extended formulation'' of a polytope $X$ $\subseteq$ $\mathbb{R}^{p}$ if
there exists a linear map $\pi$ $:$ $\mathbb{R}^{p+q}$ $\longrightarrow$
$\mathbb{R}^{p}$ such that $X $ is the image of $Q$ under $\pi$ (i.e.,
$X=\pi(Q)$; where $G\in\mathbb{R}^{m\times p}$, $H\in\mathbb{R}^{m\times q},$
and $g\in\mathbb{R}^{m}$).
\end{definition}

\begin{definition}
[``Alternate EF Definition \#2'' (Fiorini \textit{et al}. (2012))]%
\label{Extended_Polytope_Dfn3}An ``extended formulation'' of a polytope $X$
$\subseteq$ $\mathbb{R}^{p}$ is a linear system $U$ $=$ $\{(x,w)$ $\in$
$\mathbb{R}^{p+q}$ $:$ $Gx$ $+$ $Hw$ $\leq$ $g\}$ such that $x\in X$ if and
only if there exists $w\in\mathbb{R}^{q}$ such that $(x,w)\in U.$ (In other
words, $U$ is an EF of $X$ if $(x\in X\Longleftrightarrow(\exists$
$w\in\mathbb{R}^{q}:(x,w)\in U))$) (where $G$ $\in\mathbb{R}^{m\times p},$
$H\in\mathbb{R}^{m\times q},$ and $g\in\mathbb{R}^{m}$)$.$
\end{definition}

\begin{remark}
\label{EF_Dfns_Observations_Rmk}The following observations are in order with
respect to Definitions \ref{Extended_Polytope_Dfn},
\ref{Extended_Polytope_Dfn2}, and \ref{Extended_Polytope_Dfn3}:

\begin{enumerate}
\item The statement of $U$ in terms of inequality constraints only does not
cause any loss of generality, since each equality constraint can be replaced
by a pair of inequality constraints. (Yannakakis (1991, p. 442), for example)
just says that $U$ is a set of linear constraints.)

\item The statements ``$U$ is an \textit{extended formulation} of $X$'' and
``$U$ \textit{expresses} $X$'' (see Yannakakis (1991)) are equivalent.

\item The system of linear equations which specify $\pi$ in Definition
\ref{Extended_Polytope_Dfn2} must be \textit{valid} constraints for $X$ and
$U$. Hence, $X$ and $U$ can be respectively \textit{extended} by adding those
constraints to them, when trying to relate $X$ and $U$ using Definition
\ref{Extended_Polytope_Dfn2}. In that sense, Definition
\ref{Extended_Polytope_Dfn2} ``extends'' Definitions
\ref{Extended_Polytope_Dfn} and \ref{Extended_Polytope_Dfn3}.

\item All three definitions are equivalent when $G\neq\mathbf{0}$ and $U$ is
minimally-described. However, this is not true when $G=\mathbf{0,}$ as we will
show in section \ref{Ill_Dfn_Condition_Subsection} of this paper, causing a
condition of \textit{ill-definition}.

\item In the remainder of this paper, we will use the following terminologies
interchangeably: ``$A$ is an \textit{extended formulation} of $B$''; ``$A$ is
an \textit{extention }of $B$''; ``$A$ \textit{extends }$B$''; ``$B$ is
\textit{extended} by\textit{\ }$A$''. \ \ 
\end{enumerate}

\noindent$\square\medskip$
\end{remark}

\begin{remark}
\label{Conforti_Rmk}With respect to Definition \ref{Extended_Polytope_Dfn},
the following alternate definition of a projection is provided by Conforti
\textit{et al}. (2010; 2013):\medskip

Given a polyhedron $U$ $=$ $\{(x,w)$ $\in$ $\mathbb{R}^{p+q}$ $:$ $Gx$ $+$
$Hw$ $\leq$ $g\}$, its projection onto the $x$-space is $\varphi_{x}(U)=$
$\{x\in\mathbb{R}^{p}:uGx\leq ug$ for all $u\in C_{Q}\},$ where $C_{Q}%
:=\left\{  u\in\mathbb{R}^{m}:uH=\mathbf{0},u\geq\mathbf{0}\right\}  $.\medskip

Now, assume $G=\mathbf{0}$ in this and Definition \ref{Extended_Polytope_Dfn}.
Then, we would have:%
\[
\varphi_{x}(U)=\{x\in\mathbb{R}^{p}:\mathbf{0}\cdot x\leq ug\text{ for
all\ }u\in C_{Q}\}=\{x\in\mathbb{R}^{p}:ug\geq0\text{ }for\text{ all u}\in
C_{Q}\}.
\]
Hence, exactly one of the following would be true:
\begin{align*}
\varphi_{x}(U)  &  =\varnothing\text{ (if }ug<0\text{ for some }u\in
C_{Q}),\text{ or}\\[0.09in]
\varphi_{x}(U)  &  =\mathbb{R}^{p}\text{ (if }ug\geq0\text{ for all }u\in
C_{Q}).
\end{align*}
Hence, $\varphi_{x}(U)$ could not be equal to a nonempty polytope.
\ $\ \medskip$\newline $\square\medskip$
\end{remark}

\begin{definition}
[``Row-redundancy'']\label{Row_Redundancy_Dfn}Let $P:=\{x\in\mathbb{R}%
^{p}:Ax\leq a\},$ where $A\in\mathbb{R}^{m\times p}$ and $a\in\mathbb{R}^{m}.$

\begin{enumerate}
\item We say that $P$ has ``row-redundancy'' if there exists a (non-trivial)
row partitioning of $P$ with $A=\left[
\begin{tabular}
[c]{l}%
$\overline{A}_{1}$\\
$\overline{A}_{2}$%
\end{tabular}
\right]  ,$ and $a=\left[
\begin{tabular}
[c]{l}%
$\overline{a}_{1}$\\
$\overline{a}_{2}$%
\end{tabular}
\right]  $ (where $\overline{A}_{1}\in\mathbb{R}^{n\times p},$ $\overline
{A}_{2}\in\mathbb{R}^{(m-n)\times p},$ $\overline{a}_{1}\in\mathbb{R}^{n}$ and
$\overline{a}_{2}\in\mathbb{R}^{(m-n)})\ $such that one of the following
conditions is true:

\begin{enumerate}
\item $\left(  \overline{x}\in P\right)  \Longleftrightarrow\left(
\overline{x}\in\{x\in\mathbb{R}^{p}:\overline{A}_{1}x\leq\overline{a}%
_{1}\}\right)  ,$ or \medskip

\item $\left(  \overline{x}\in P\right)  \Longleftrightarrow\left(
\overline{x}\in\{x\in\mathbb{R}^{p}:\overline{A}_{2}x\leq\overline{a}%
_{2}\}\right)  $.
\end{enumerate}

\item We say that the constraints $\overline{A}_{2}x\leq\overline{a}_{2}$ are
``redundant'' for $\{x\in\mathbb{R}^{p}:\overline{A}_{1}x\leq\overline{a}%
_{1}\}$ if Condition ($1.a$) is true. Similarly, we say that the constraints
$\overline{A}_{1}x\leq\overline{a}_{1}$ are ``redundant'' for $\{x\in
\mathbb{R}^{p}:\overline{A}_{2}x\leq\overline{a}_{2}\}$ if Condition ($1.b$)
is true.$\medskip$
\end{enumerate}
\end{definition}

\begin{definition}
[``Column-redundancy'']\label{Column_Redundancy_Dfn}Let $P:=\{x\in
\mathbb{R}^{p}:Ax\leq a\},$ where $A\in\mathbb{R}^{m\times p}$ and
$a\in\mathbb{R}^{m}.$ Let $x$ denote the descriptive variables of $P.$ Let
$\left[
\begin{tabular}
[c]{l}%
$\overline{x}_{1}$\\
$\overline{x}_{2}$%
\end{tabular}
\right]  $ be a (non-trivial) partitioning of $x$, where $\overline{x}_{1}%
\in\mathbb{R}^{q},$ and $\overline{x}_{2}\in\mathbb{R}^{(p-q)}$.

\begin{enumerate}
\item We say that $P$ has ``column-redundancy'' if one of the following
conditions is true:\medskip

\begin{enumerate}
\item $\exists\left(  B_{1},b_{1}\right)  \in\mathbb{R}^{n\times q}%
\times\mathbb{R}^{n}:\medskip$

$\left(  \left[
\begin{tabular}
[c]{l}%
$\overline{\overline{x}}_{1}$\\
$\overline{\overline{x}}_{2}$%
\end{tabular}
\right]  \in Ext(P)\Longrightarrow\overline{\overline{x}}_{1}\in Ext\left(
\left\{  x\in\mathbb{R}^{q}:B_{1}x\leq b_{1}\right\}  \right)  \right.
,\ $and$\ \medskip$

$\left.  \overline{\overline{x}}_{1}\in Ext\left(  \left\{  x\in\mathbb{R}%
^{q}:B_{1}x\leq b_{1}\right\}  \right)  \Longrightarrow\exists\overline
{\overline{x}}_{2}\in\mathbb{R}^{(p-q)}:\left[
\begin{tabular}
[c]{l}%
$\overline{\overline{x}}_{1}$\\
$\overline{\overline{x}}_{2}$%
\end{tabular}
\right]  \in Ext(P)\right)  \medskip$

(where $1\leq n\leq m),$ or\medskip\medskip

\item $\exists\left(  B_{2},b_{2}\right)  \in\mathbb{R}^{n\times q}%
\times\mathbb{R}^{n}:\medskip$

$\left(  \left[
\begin{tabular}
[c]{l}%
$\overline{\overline{x}}_{1}$\\
$\overline{\overline{x}}_{2}$%
\end{tabular}
\right]  \in Ext(P)\Longrightarrow\overline{\overline{x}}_{2}\in Ext\left(
\left\{  x\in\mathbb{R}^{(p-q)}:B_{2}x\leq b_{2}\right\}  \right)  \right.
$,\ and \medskip

$\left.  \overline{\overline{x}}_{2}\in Ext\left(  \left\{  x\in
\mathbb{R}^{(p-q)}:B_{2}x\leq b_{2}\right\}  \right)  \Longrightarrow
\exists\overline{\overline{x}}_{1}\in\mathbb{R}^{q}:\left[
\begin{tabular}
[c]{l}%
$\overline{\overline{x}}_{1}$\\
$\overline{\overline{x}}_{2}$%
\end{tabular}
\right]  \in Ext(P)\right)  $ \medskip

(where $1\leq n\leq m).$
\end{enumerate}

\item We say that the variables $\overline{x}_{2}$ are ``redundant'' for
$\left\{  x\in\mathbb{R}^{q}:B_{1}x\leq b_{1}\right\}  $ when Condition
($1.a$) is true. Similarly, we say that variables $\overline{x}_{1}$ are
``redundant'' for $\left\{  x\in\mathbb{R}^{(p-q)}:B_{2}x\leq b_{2}\right\}  $
when Condition ($1.b$) is true.)\medskip
\end{enumerate}
\end{definition}

\begin{definition}
[``Minimally-described'' polytope]\label{Minimal_Description_Dfn}We say that a
polyhedron $P$ is ``minimally-described,'' or that (the statement of) $P$ is
``minimal,'' if $P $ has no \textit{row-redundancy} and no
\textit{column-redundancy}.
\end{definition}

\begin{assumption}
\label{Minimality_Assumption}In the remainder of this paper, with respect to
Definitions \ref{Extended_Polytope_Dfn}, \ref{Extended_Polytope_Dfn2}, and
\ref{Extended_Polytope_Dfn3}, we will assume (implicitly) that $U$ is
\textit{minimally-described} whenever we will be considering (or referring to)
the case in which $G\neq\mathbf{0.\medskip}$
\end{assumption}

Observation \ref{EF_Dfns_Observations_Rmk}.4 and Remark \ref{Conforti_Rmk}
above are the key point in the concept of \textit{ill-definition} of an
\textit{extended formulation} which occurs in the special of $G=\mathbf{0}$ in
Definitions \ref{Extended_Polytope_Dfn}, \ref{Extended_Polytope_Dfn2}, and
\ref{Extended_Polytope_Dfn3}. This allows for a \textit{barrier} to be removed
by using an alternate formulation for a given combinatorial optimization
problem (COP) at hand.\medskip

\subsection{The ``Alternate TSP Polytope'': An example of non-exponential
abstraction of TSP\ tours\label{AP_Formulation_of_TSP_Polytope_Section}}

We now introduce a non-exponential LP model which correctly abstracts TSP
tours. We will use this model as an illustrative example for our discussions
in section \ref{Ill_Definition_Section} addressing \textit{extension}
relations to the \textit{Standard TSP Polytope}, after we have formalized our
\textit{ill-definition} conditions in terms of the differences between
Definitions \ref{Extended_Polytope_Dfn}, \ref{Extended_Polytope_Dfn2}, and
\ref{Extended_Polytope_Dfn3}.

\begin{theorem}
\label{Correspondence_AP&Tours} \ Consider the TSP defined on the set of
cities $\Omega:=\{1,\ldots,n\}$. Assume city ``$1$'' has been designated as
the staring and ending point of the travels. Let $S:=\{1,\ldots,n-1\}$ denote
the \textit{times-of-travel} to cities ``$2$'' through ``$n$.'' Then, there
exists a one-to-one correspondence between TSP tours and extreme points of
\[
AP:=\left\{  w\in\mathbb{R}_{\mathbb{\nless}}^{(n-1)^{2}}:\sum\limits_{t\in
S}w_{i,t}=1\text{ \ }\forall i\in(\Omega\backslash\{1\});\text{ \ }%
\sum\limits_{i\in(\Omega\backslash\{1\})}w_{i,t}=1\text{ \ }\forall t\in
S\right\}  .
\]
\end{theorem}

\begin{proof}
Using the assumption that node $1$ is the starting and ending point of travel,
it is trivial to construct a unique TSP tour from a given extreme point of
$AP,$ and vice versa (i.e., it is trivial to construct a unique extreme point
of $AP$ from a given TSP tour). \ \medskip
\end{proof}

\begin{corollary}
\label{Contradiction_Coroll}It follows directly from Theorem
\ref{Correspondence_AP&Tours} that $AP$ is a contradiction of any notion
whereby it is ``impossible'' to abstract the TSP polytope into a linear
program of polynomial size, since $AP$ clearly has polynomial (linear) size
and it is a well-established fact that $AP$ is integral (see Burkard
\textit{et al}. (2009)).
\end{corollary}

\begin{remark}
\label{Alternate_Abstractions_Rmk}Hence, alternate abstractions of the TSP
optimization problem which may or may not involve the \textit{Standard TSP
Polytope} are possible. For example, in the ``standard'' (i.e.,
\textit{Standard TSP Polytope}-based) abstraction of the TSP optimization
problem, the cost function is trivial to develop. The challenge is to come up
with linear constraints so that the extreme points of the induced polytope are
TSP tours. On the other hand, in an abstraction based on $AP,$ the
representation of the TSP tours is a straightforward matter (since the tours
are abstracted into assignment problem (bipartite matching) solutions). The
challenge is to find appropriate costs to apply for these thus-abstracted TSP
tours. Clearly, this challenge of coming up with a cost function is not within
the scope of EF developments for the \textit{Standard TSP Polytope}, since it
does not involve that polytope. Examples of how this challenge can be met are
described in Diaby (2007) and in Diaby and Karwan (2012), respectively. Also,
clearly, from an overall perspective, one cannot reasonably equate an
``impossibility'' of meeting the challenge in one of the two abstractions
(i.e., the ``\textit{Standard TSP Polytope}-based'' and ``$AP$-based''
abstractions) with an ``impossibility'' of meeting the challenge in the other.
\ \ $\square$
\end{remark}

\begin{remark}
\label{Alternate_TSP_Polytope_Rmk}More formally, clearly, $AP$\ is also a TSP
polytope, since its extreme points correspond to TSP tours. According to the
\textit{Minkowski-Weyl Theorem} (Minskowski (1910); Weyl (1935); see also
Rockafellar (1997, pp.153-172)), every polytope can be equivalently described
as the intersection of hyperplanes ($\mathcal{H}$-representation/external
description) or as a convex combination of (a finite number of) vertices
($\mathcal{V}$-representation/internal description). The \textit{Standard TSP
Polytope} is stated in terms of its $\mathcal{V}$-representation. No
polynomial-sized $\mathcal{H}$-representation of it is known. On the other
hand, $AP$ is stated in terms of its $\mathcal{H}$-representation (which is
well-known to be of (low-degree) polynomial size (see Burkard \textit{et al}.
(2009)), but it is trivial to state its $\mathcal{V}$-representation also. The
vertices of $AP$\ are assignment problem solutions, whereas the vertices of
the \textit{Standard TSP Polytope} are Hamiltonian cycles.\textit{\ }Hence,
even though the extreme points of $AP$ and those of the \textit{Standard TSP
Polytope} respectively correspond to TSP tours, the two sets of extreme points
are different kinds of mathematical objects, with un-related mathematical
characterizations. Hence, there does not exist any \textit{a priori}
mathematical relation between $AP$\ and the \textit{Standard TSP Polytope. }In
other words, $AP$\ and the \textit{Standard TSP Polytope} are simply alternate
abstractions of TSP tours. Or, put another way, $AP$ is (simply) an alternate
TSP polytope from the \textit{Standard TSP Polytope}, and vice versa (i.e.,
that the \textit{Standard TSP Polytope} is (simply) an alternate TSP polytope
from $AP$). \ \ $\square$
\end{remark}

\begin{definition}
[``Alternate TSP Polytope'']\label{Alternate_Polytope_Dfn}We refer to $AP$ as
the ``Alternate TSP Polytope.''
\end{definition}

\section{\textit{Ill-definition c}ondition for ``Extended
Formulations''\label{Ill_Definition_Section}}

\subsection{Polytope Augmentation\label{Polytope_Augmentation_Subsection}}

\begin{definition}
[``Class of variables'']\label{Class_Of_Variables_Dfn}We refer to a set of
variables which model a given aspect of a problem, as a ``class of
variables.'' The \textit{traditional }$\mathit{x}$\textit{-variables} for
example, would constitute one \textit{class of variables in a TSP model}, as
they represent (single) ``travel legs'' in the TSP. Similarly, the $y$- and
$z$-variables used in the models of Diaby (2007) and Diaby and Karwan (2012),
respectively, would constitute two distinct \textit{classes of variables,}
with the $y$-variables modeling doublets of ``travel legs'' in the TSP, and
the $z$-variables modeling triplets of ``travel legs'' in the TSP.
\end{definition}

\begin{assumption}
In the remainder of this paper, we will assume (without loss of generality)
that a given \textit{class of variables} is denoted by the same symbol in all
of the models in which it is used. That is, we will assume that the same
notation (whatever that may be) will be used to designate the
\textit{traditional }$\mathit{x}$\textit{-variables} for example, in every
model in which these variables are used.
\end{assumption}

\begin{definition}
[``Independent spaces'']\label{Independent_Spaces_Dfn}Let $x\in\mathbb{R}%
^{p}\ (p\in\mathbb{N}_{+}) $ and $w\in\mathbb{R}^{q}$ $(q\in\mathbb{N}_{+})$
be the vectors of descriptive variables for two polyhedra in $\mathbb{R}^{p}$
and $\mathbb{R}^{q},$ respectively.

\begin{enumerate}
\item We say that $x$ and $w$ (or that the polyhedra) are in ``independent
spaces\textit{''} if $x$ and $w$ do not have any \textit{class of variables}
in common. That is, we say that $x$ and $w$ are in ``independent
spaces\textit{''} if the following conditions holds:

\begin{enumerate}
\item $x$ cannot be partitioned as $x=\dbinom{\overline{x}}{w};$

\item $w$ cannot be partitioned as $w=\dbinom{\overline{w}}{x};$ and

\item $\forall m\in\mathbb{N}_{+}:$ $m<\min\{p,q\},$ $\nexists(\overline
{x},\overline{w},v)\in\mathbb{R}^{(p-m)}\times\mathbb{R}^{(q-m)}%
\times\mathbb{R}^{m}:(x$ and $w$ can be respectively partitioned as
$x=\dbinom{\overline{x}}{v}$ and $w=\dbinom{\overline{w}}{v}),$ where $v$
denotes a given \textit{class of variables} for the problem at hand.
\end{enumerate}

\item We will say that $x$ and $w$ (or that the polyhedra they respectively
describe) ``overlap'' if $x$ and $w$ have one or more \textit{classes of
variables} in common.
\end{enumerate}
\end{definition}

Regan and Lipton (2013) remarked that all polytopes may be viewed, in a
degenerate way, as being part of one overall multi-dimensional space. The
following alternate (and equivalent) definition of ``independent spaces'' is
therefore useful in further clarifying the notion.

\begin{definition}
[Alternate definition of ``Independent spaces'']%
\label{Independent_Spaces_Alternate_Dfn}Let $P$ and $Q$ be polytopes in
$\mathbb{R}^{p+q},$ with descriptive variables $(x,y)\in\mathbb{R}^{p}%
\times\mathbb{R}^{q}.$ We say that $P$ and $Q$ are in ``independent spaces''
iff exactly one of the following two conditions holds:

\begin{enumerate}
\item $\left\{  x\in\mathbb{R}^{p}:\left(  \exists y\in\mathbb{R}^{q}:(x,y)\in
P\right)  \right\}  =\mathbb{R}^{p}$ \ and \ $\left\{  y\in\mathbb{R}%
^{q}:\left(  \exists x\in\mathbb{R}^{p}:(x,y)\in Q\right)  \right\}
=\mathbb{R}^{q};$

\item $\left\{  y\in\mathbb{R}^{q}:\left(  \exists x\in\mathbb{R}^{p}:(x,y)\in
P\right)  \right\}  =\mathbb{R}^{q}$ \ and \ $\left\{  x\in\mathbb{R}%
^{p}:\left(  \exists y\in\mathbb{R}^{q}:(x,y)\in Q\right)  \right\}
=\mathbb{R}^{p}.\smallskip$
\end{enumerate}
\end{definition}

\begin{example}
\label{Indep_Spaces_Illustr}Definitions \ref{Independent_Spaces_Dfn} and
\ref{Independent_Spaces_Alternate_Dfn} can be illustrated as follows.

\begin{itemize}
\item Assume $x\in\mathbb{R}^{2}$ and $y\in\mathbb{R}^{2}$ refer to different
\textit{classes of variables} in a modeling context at hand.

\item Let $x$ and $y$ be the descriptive variables for Polytopes $P$ and $Q$
respectively, with:%
\begin{align*}
P  &  :=\{x\in\mathbb{R}^{2}:x_{1}-x_{2}\geq6;\text{ \ }0\leq x_{1}%
\leq6;\text{ \ }0\leq x_{2}\leq5\};\\[0.09in]
Q  &  :=\{y\in\mathbb{R}^{2}:y_{1}+y_{2}=6;\text{ \ }y_{1}\geq1.5;\text{
\ }y_{2}\geq0\}.
\end{align*}

\item Clearly, $P$ and $Q$ are in $\mathbb{R}^{4}$ in a degenerate sense, respectively.

\item However:

\begin{itemize}
\item $P$ and $Q$ are \textit{independent spaces a}ccording to Definition
\ref{Independent_Spaces_Dfn} directly;

\item $P$ and $Q$ can be respectively re-written as:\newline $P^{\prime
}=\left\{  (x,y)\in\mathbb{R}^{2}\times\mathbb{R}^{2}:\mathbf{A}%
x+\mathbf{0}\cdot y\leq\mathbf{a}\right\}  $; where $\mathbf{A}=\left[
\begin{array}
[c]{cc}%
-1 & 1\\
-1 & 0\\
1 & 0\\
0 & -1\\
0 & 1
\end{array}
\right]  ,$ $\mathbf{a}=\left[
\begin{array}
[c]{c}%
-6\\
0\\
6\\
0\\
5
\end{array}
\right]  ,$ \ and \newline $Q^{\prime}:=\left\{  (x,y)\in\mathbb{R}^{2}%
\times\mathbb{R}^{2}:\mathbf{0}\cdot x+\mathbf{B}y\leq\mathbf{b}\right\}  $;
where $\mathbf{B}=\left[
\begin{array}
[c]{cc}%
1 & 1\\
-1 & -1\\
-1 & 0\\
0 & -1
\end{array}
\right]  ,$ $\mathbf{b}=\left[
\begin{array}
[c]{c}%
6\\
-6\\
-1.5\\
0
\end{array}
\right]  ,$

so that:\medskip\medskip\newline $\left\{  y\in\mathbb{R}^{2}:\left(  \exists
x\in\mathbb{R}^{2}:(x,y)\in P^{\prime}\right)  \right\}  =\mathbb{R}^{2}$
\ and \ $\left\{  x\in\mathbb{R}^{2}:\left(  \exists y\in\mathbb{R}%
^{2}:(x,y)\in Q^{\prime}\right)  \right\}  =\mathbb{R}^{2}.\medskip\medskip$

Hence, $P^{\prime}$ and $Q^{\prime}$ (and therefore, $P$ and $Q)$ are
\textit{independent spaces }according to Definition
\ref{Independent_Spaces_Alternate_Dfn}.
\end{itemize}
\end{itemize}

\noindent$\square$ \ 
\end{example}

\begin{definition}
[``Polyhedron augmentation'']\label{Polytope_Augmentation_Dfn}Let $X$ be a
non-empty polyhedron described in terms of variables $x\in\mathbb{R}^{p}$. Let
$\overline{X}$ be a polyhedron the description of which consists of the
constraints of $X,$ plus additional variables and constraints that are not
used in the description of $X$. We will say that $\overline{X}$ is an
``augmentation'' of $X$ (or that $\overline{X}$ ``augments'' $X)$ if the
problem of optimizing any given linear function of $x$ over $X,$ is equivalent
to the problem of optimizing that linear function over $\overline{X}$. In
other words, let $x\in\mathbb{R}^{p}$ and $y\in\mathbb{R}^{q}$ be vectors of
variables in \textit{independent spaces}. Let $X:=\{x\in\mathbb{R}^{p}:Ax\leq
a\}\neq\varnothing,$ and $\overline{X}:=\{(x,y)\in\mathbb{R}^{p+q}:Ax\leq a;$
$Bx+Cy\leq b\}\neq\varnothing$ (where $A\in$ $\mathbb{R}^{k\times p},$
$a\in\mathbb{R}^{k},$ $B\in$ $\mathbb{R}^{m\times p},$ $C\in$ $\mathbb{R}%
^{m\times q},$ and $b\in\mathbb{R}^{m}$). We say that $\overline{X}$
\textit{augments} $X$ if $(\forall x\in X,$ $\exists y\in\mathbb{R}^{q}:$
$(x,y)\in\overline{X}).$\ 
\end{definition}

\begin{remark}
\label{Polytope_Augmentation_Rmk1}With respect to Definition
\ref{Polytope_Augmentation_Dfn}:

\begin{enumerate}
\item The additional variables and constraints of $\overline{X}$
are\textit{\ redundant} for $X$ (see Definitions \ref{Row_Redundancy_Dfn} and
\ref{Column_Redundancy_Dfn}); \ 

\item The optimization problem over $\overline{X}$ may not be equivalent to
the optimization problem over $X$, if the objective function in the problem
over $\overline{X}$ is changed from that of $X$ to include non-zero terms of
the new variables;

\item Every augmentation of $X$ is an \textit{extended formulation} of $X$,
but the converse is not true (since an \textit{extended formulation} of $X$
need not include the constraints of $X$ explicitly);

\item The polyhedral set associated to an optimization problem is equivalent
to all of its \textit{augmentations} respectively, provided the expression of
the objective function of the optimization problem is not changed;

\item In the discussions to follow we will assume (w.l.o.g.) that the
objective function is not changed when new variables and constraints are added
to an optimization problem. Hence, in the discussions to follow, we will not
distinguish between a polyhedral set and the optimization problem to which it
is associated, except for where that causes confusion.
\end{enumerate}

\noindent$\square$
\end{remark}

\begin{example}
\label{Polytope_Augmentation_Example} \ We illustrate Definition
\ref{Polytope_Augmentation_Dfn} and \ Remark \ref{Polytope_Augmentation_Rmk1}
as follows. \medskip

Let:

\begin{description}
\item $(i)$ $x\in\mathbb{R}^{p}$ and $y\in\mathbb{R}^{q}$ be variables in
\textit{independent spaces};

\item $(ii)$ $X:=\{x\in\mathbb{R}^{p}:Ax\leq a\};$

\item $(iii)$ $L:=\{(x,y)\in\mathbb{R}^{p+q}:Bx+Cy\leq c\}$;

\item $(iv)$ $Y:=\{y\in\mathbb{R}^{q}:Dy\leq d\};$

\item $(v)$ $K_{1}:=\left\{  (x,y)\in\mathbb{R}^{p+q}:\left[
\begin{array}
[c]{cc}%
A & \mathbf{0}\\
B & C
\end{array}
\right]  \left[
\begin{array}
[c]{c}%
x\\
y
\end{array}
\right]  \leq\left[
\begin{array}
[c]{c}%
a\\
c
\end{array}
\right]  \right\}  ;$

\item $(vi)$ $K_{2}:=\left\{  (x,y)\in\mathbb{R}^{p+q}:\left[
\begin{array}
[c]{cc}%
B & C\\
\mathbf{0} & D
\end{array}
\right]  \left[
\begin{array}
[c]{c}%
x\\
y
\end{array}
\right]  \leq\left[
\begin{array}
[c]{c}%
c\\
d
\end{array}
\right]  \right\}  ;$

\item $(vii)$ $K_{3}:=\left\{  (x,y)\in\mathbb{R}^{p+q}:\left[
\begin{array}
[c]{cc}%
A & \mathbf{0}\\
B & C\\
\mathbf{0} & D
\end{array}
\right]  \left[
\begin{array}
[c]{c}%
x\\
y
\end{array}
\right]  \leq\left[
\begin{array}
[c]{c}%
a\\
c\\
d
\end{array}
\right]  \right\}  $. $\ $
\end{description}

\noindent(Where: $\ A\in\mathbb{R}^{k\times p};$ $a\in\mathbb{R}^{k};$
$B\in\mathbb{R}^{m\times p};$ $C\in\mathbb{R}^{m\times q};$ $c\in
\mathbb{R}^{m};$ $D\in\mathbb{R}^{l\times q},$ $d\in\mathbb{R}^{l}%
$)$.\smallskip$

Assume:

\begin{description}
\item $(vii)$ $A\neq\mathbf{0},$ $B\neq\mathbf{0},$ $C\neq\mathbf{0},$
$D\neq\mathbf{0};$

\item $(viii)$ $B$ cannot be partitioned as $B=\left[
\begin{array}
[c]{c}%
A\\
\overline{B}%
\end{array}
\right]  ;$

\item $(ix)$ $\ C$ cannot be partitioned as $C=\left[
\begin{array}
[c]{c}%
\overline{C}\\
D
\end{array}
\right]  ;$

\item $(x)$ \ \ The constraints of $L$ are redundant for $X$ and $Y$.\smallskip
\end{description}

Then:

\begin{description}
\item $(xi)$ $K_{1}$ is an \textit{augmentation} of $X$, but not of $L$, nor
of $Y$;

\item $(xii)$ $K_{2}$ is an \textit{augmentation} of $Y$, but not of $L,$ nor
of $X$;

\item $(xiii)$ $K_{3}$ is an \textit{augmentation} of $X$ and $Y$, but not of
$L$;

\item $(xiv)$ $K_{1}$ is an \textit{extended formulation} of $X$, but not of
$L,$ and may or may not be for $Y;$

\item $(xv)$ $K_{2}$ is an \textit{extended formulation} of $Y$, but not of
$L,$ and may or may not be for $X$;

\item $(xvi)$ $K_{3}$ is an \textit{extended formulation} of $X$ and $Y$, but
not of $L$;

\item $(xvii)$ $L$ is not an \textit{augmentation} of $X$ nor of $Y;$

\item $(xviii)$ $L$ may or may not be an \textit{extended formulation} of $X$;

\item $(xix)$ $L$ may or may not be an \textit{extended formulation} of $Y.$
\end{description}

\noindent$\square$ \ 
\end{example}

\begin{remark}
\label{Polytope_Augmentation_Rmk2}In reference to the developments above:

\begin{enumerate}
\item We will refer to the constraints of $L$ as the ``linking constraints''
(for $X$ and $Y$)\ in $K_{3},$ regardless of whether or not the constraints of
$L$ are redundant for $X$ and $Y$;

\item If $X$ and $Y$ are alternative correct abstractions of the requirements
of some (same) given optimization problem, then there may or may not exist $B$
and $C$ such that $((x,y)\in K_{1}$ $\Longrightarrow y\in Y)$\ and $((x,y)\in
K_{2}$ $\Longrightarrow x\in X)$. This is exemplified by the \textit{Alternate
TSP Polytope} relative to the \textit{Standard TSP Polytope};

\item If there exist $B$ and $C$ such that $((x,y)\in K_{1}$ $\Longrightarrow
y\in Y)$\ and $((x,y)\in K_{2}$ $\Longrightarrow x\in X),$ then it must be
possible to attach meanings to $x$ and $y$, so that $X$ and $Y$ are
alternative correct abstractions of the requirements of some (same) given
optimization problem. This is exemplified by the LP models of the TSP in Diaby
(2007) and in Diaby and Karwan (2012), respectively, relative to the
\textit{Alternate TSP Polytope}, or relative to the \textit{Standard TSP
Polytope};

\item The main point of our developments in section
\ref{Ill_Definition_Section} below will be to show that there exists no
\textit{we\textit{ll-defined}} (non-ambiguous, meaningful) \textit{extended
formulation} relationship between $X$ and $Y$.

In particular, we will show that
\[
(\exists(B,C):(x,y)\in K_{1}\Longrightarrow y\in Y)\nRightarrow(X\text{ is a
\textit{we\textit{ll-defined}} \textit{extended formulation} of }Y)\text{,\ }%
\]
and that similarly,
\[
(\exists(B,C):(x,y)\in K_{2}\Longrightarrow x\in X)\nRightarrow(Y\text{ is a
\textit{we\textit{ll-defined}} \textit{extended formulation} of }X)\text{.\ }%
\]
For example, there exist linear transformations which allow for points of the
\textit{Alternate TSP Polytope} to be ``retrieved'' from (given) solutions of
the TSP LP models in Diaby [2007] and Diaby and Karwan [2012]. Note however,
that the ``retrieval'' of points of the \textit{Standard TSP\ Polytope} can be
accomplished only through the use of ``implicit'' information (about TSP node
``$1$'' specifically) that is outside the scope of the TSP LP models per se.
Hence, the TSP LP models can be \textit{we\textit{ll-defined} extended
formulations} of the \textit{Standard TSP Polytope} only if they are
\textit{we\textit{ll-defined} extended formulations of }the \textit{Alternate
TSP Polytope}, which would seem to suggest that the \textit{Alternate TSP
Polytope} must be a \textit{we\textit{ll-defined} extended formulation} of the
\textit{Standard TSP polytope}. We do not believe such a suggestion is the
intent of any \textit{extended formulations }work. However, we believe the
definitions of an ``extended formulation'' must be properly interpreted in
order for them not to lead to such conclusions. Specifically, using the notion
of \textit{augmentation} discussed in this section, we will show in the next
section (section \ref{Ill_Definition_Section}) that the notion of an EF can
become \textit{ill-defined} (and thereby lose its meaningfulness) when the
polytopes being related are expressed in coordinate systems that are
independent of each other.
\end{enumerate}

\noindent$\square$\medskip
\end{remark}

We will now discuss two results which will be helpful subsequently in showing
the differences between the case of polytopes in \textit{overlapping spaces}
and the case of polytopes in \textit{independent spaces}, as pertains to
\textit{extension} relationships.

\begin{theorem}
\label{Overlapg_Polytopes_Thm} Let $P_{1}$ and $P_{2}$ be non-empty,
\textit{minimally-described} polytopes in \textit{overlapping spaces} with the
set of the descriptive variables of $P_{1}$ included in that of $P_{2}$. An
\textit{augmentation} of $P_{2}$ is an \textit{extended formulation }of
$P_{1}$ if and only if $P_{2}$ is an \textit{extended formulation }of $P_{1},$
according to Definitions \ref{Extended_Polytope_Dfn},
\ref{Extended_Polytope_Dfn2}, and \ref{Extended_Polytope_Dfn3}, respectively.\medskip

In other words, let:%
\[
P_{1}:=\left\{  x\in\mathbb{R}^{q_{1}}:A_{1}x\leq a_{1}\right\}  \text{ (where
}A_{1}\in\mathbb{R}^{r_{1}\times q_{1}};\text{ }a_{1}\in\mathbb{R}^{r_{1}});
\]%
\[
P_{2}:=\left\{  (x,u)\in\mathbb{R}^{q_{1}+q_{2}}:A_{2}x+Bu\leq b\right\}
\text{ }(\text{where: }A_{2}\in\mathbb{R}^{r_{2}\times q_{1}};\text{ }%
B\in\mathbb{R}^{r_{2}\times q_{2}};\text{ }b\in\mathbb{R}^{r_{2}}).
\]

Assume $A_{1}\neq\mathbf{0},$ $A_{2}\neq\mathbf{0},$ $P_{1}\neq\varnothing,$
and $P_{2}\neq\varnothing.$ Then, an arbitrary \textit{augmentation}, $P_{3},$
of $P_{2}$ is an \textit{extended formulation} of $P_{1}$ if and only if
$P_{2}$ is an \textit{extended formulation} of $P_{1},$ according to
Definitions \ref{Extended_Polytope_Dfn}, \ref{Extended_Polytope_Dfn2}, and
\ref{Extended_Polytope_Dfn3}, respectively.
\end{theorem}

\begin{proof}
First, note that Definitions \ref{Extended_Polytope_Dfn},
\ref{Extended_Polytope_Dfn2}, and \ref{Extended_Polytope_Dfn3} are equivalent
to one another with respect to \textit{extension} relations for $P_{1}$,
$P_{2}$, and $P_{3}$ (see Remark \ref{EF_Dfns_Observations_Rmk}.4). Hence, it
is sufficient to prove the theorem for the \textit{standard definition}
(Definition \ref{Extended_Polytope_Dfn}). \medskip

$P_{3}$ can be written as:%
\begin{align*}
&  P_{3}:=\left\{  (x,u,v)\in\mathbb{R}^{q_{1}+q_{2}+q_{3}}:A_{2}x+Bu\leq
b;\text{ }A_{3}x+Cu+Dv\leq c\right\}  \text{ \ }(\text{where}\text{: }A_{3}%
\in\mathbb{R}^{r_{3}\times q_{1}};\\
&  C\in\mathbb{R}^{r_{3}\times q_{2}};D\in\mathbb{R}^{r_{3}\times q_{3}}%
;c\in\mathbb{R}^{r_{3}};\text{ }A_{3}x+Cu+Dv\leq c\text{ are \textit{redundant
for} }P_{2}).
\end{align*}

($A_{3}x+Cu+Dv\leq c)$ \textit{redundant for} $P_{2}$ $\Longrightarrow:$%
\begin{equation}
\forall(x,u)\in P_{2},\text{ }\exists v\in\mathbb{R}^{q_{3}}:(x,u,v)\in P_{3}.
\label{OLPs(a)}%
\end{equation}

(\ref{OLPs(a)}) $\Longrightarrow:$%
\begin{equation}
\{x\in\mathbb{R}^{q_{1}}:(\exists(u,v)\in\mathbb{R}^{q_{1}\mathbb{+}q_{2}%
}:(x,u,v)\in P_{3})\}=\{x\in\mathbb{R}^{q_{1}}:(\exists u\in\mathbb{R}^{q_{2}%
}:(x,u)\in P_{2})\}. \label{OLPs(b)}%
\end{equation}

(\ref{OLPs(b)}) $\Longrightarrow:$%
\[
\left(  \varphi_{x}(P_{3})=P_{1}\right)  \Longleftrightarrow\left(
\varphi_{x}(P_{2})=P_{1}\right)  .
\]
\medskip
\end{proof}

We will show in the next section (section \ref{Ill_Dfn_Condition_Subsection})
that Theorem \ref{Overlapg_Polytopes_Thm} does not hold\ for polyhedra which
are stated in \textit{independent spaces }(such as $P$ and $Q$ in Example
\ref{Indep_Spaces_Illustr}, or $X$ and $Y$ in Example
\ref{Polytope_Augmentation_Example}), and that, as indicated in Remark
\ref{Polytope_Augmentation_Rmk2} above, there cannot exist any
we\textit{ll-defined }(non-ambiguous, meaningful) \textit{extension}
relationship between such polytopes.

\subsection{Ill-definition condition for
EF's\label{Ill_Dfn_Condition_Subsection}}

Referring back to the \textit{standard} and \textit{alternate}
\textit{definitions} of \textit{extended formulations} (i.e., Definitions
\ref{Extended_Polytope_Dfn}, \ref{Extended_Polytope_Dfn2}, and
\ref{Extended_Polytope_Dfn3}, respectively), it is easy to verify (as
indicated in Remark \ref{EF_Dfns_Observations_Rmk}.4) that these three
definitions are equivalent when $G\neq\mathbf{0}$ (with $U$
\textit{minimally-described}). In other words, one can easily verify that
provided $G\neq\mathbf{0,}$ $U$ is an \textit{extended formulation} of $X$
according to one of these definitions if and only if $U$ is an
\textit{extended formulation} of $X$ according to the other definitions.
However, this is not true when $G=\mathbf{0.}$ \ \medskip

A basic intuition of Definitions \ref{Extended_Polytope_Dfn},
\ref{Extended_Polytope_Dfn2}, and \ref{Extended_Polytope_Dfn3} is that if the
projection of $U$ onto $x$\textit{-}space is equal to $X$, then the
description of $X$ must be implicit in a constraint set of the form:%
\[
Gx\leq\overline{g}_{w}.
\]

Hence, the notion of an \textit{extended formulation} can become
\textit{ill-defined} when $G=\mathbf{0}$ (i.e., when $U$ and $X$ are
\textit{independent spaces})$\mathbf{.}$ In essence, to put it roughly, there
is ``nothing'' in the statement of $U$ for the constraints of $X$ to be
``implicit in''\ (in $U)$ when $G=\mathbf{0.}$ Indeed, as we will show in the
discussion to follow, when $G=\mathbf{0,}$ Definition
\ref{Extended_Polytope_Dfn2} can become contradictory of Definitions
\ref{Extended_Polytope_Dfn} and \ref{Extended_Polytope_Dfn3}, resulting in an
\textit{ill-definition} (ambiguity) condition.\medskip

The theorem below shows that there exist no \textit{extension}
relations\textit{\ }between polytopes stated in \textit{independent spaces}
according to the \textit{standard definition }(Definition
\ref{Extended_Polytope_Dfn}), or the \textit{second alternate definition}
(Definition \ref{Extended_Polytope_Dfn3})\textit{\ }of \textit{extended
formulations}.

\begin{theorem}
\label{Independent_Spaces_Thm}Polytopes described in \textit{independent
spaces} cannot be \textit{extended formulations} of each other according to
the \textit{standard definition} (Definition \ref{Extended_Polytope_Dfn}) or
the \textit{second alternate definition} (Definition
\ref{Extended_Polytope_Dfn3}) of \textit{extended formulations}.\medskip

In other words, let $X$ $:=$ $\left\{  x\in\mathbb{R}^{p}:Ax\leq a\right\}  $
and $U:=\left\{  w\in\mathbb{R}^{q}:Hw\leq h\right\}  $ be (non-empty)
polytopes in \textit{independent spaces} (where: $A\in\mathbb{R}^{m\times p};$
$a\in\mathbb{R}^{m};$ $H\in\mathbb{R}^{n\times q};$ $h\in\mathbb{R}^{n}%
$)\textit{.} Then:

\begin{description}
\item $(i)$\textit{\ }$U$\ cannot be an \textit{extended formulation} of $X$
(and vice versa) according to Definition \ref{Extended_Polytope_Dfn};

\item $(ii)$\textit{\ }$U$\ cannot be an \textit{extended formulation} of $X$
(and vice versa) according to Definition \ref{Extended_Polytope_Dfn3}.
\end{description}
\end{theorem}

\begin{proof}
We will show that $U$ cannot be an \textit{extended formulation} of $X$
according to Definitions \ref{Extended_Polytope_Dfn} and
\ref{Extended_Polytope_Dfn3}, respectively. The proofs that $X$ cannot be an
\textit{extended formulation} of $U$ according to the definitions (Definitions
\ref{Extended_Polytope_Dfn} and \ref{Extended_Polytope_Dfn3}, respectively)
are similar and will be therefore omitted.\medskip

Let $U$ be re-stated in $\mathbb{R}^{p+q}$ as:%
\begin{equation}
U^{\prime}:=\left\{  (x,w)\in\mathbb{R}^{p+q}:\mathbf{0}\cdot x+Hw\leq
h\right\}  . \label{EFThm1a}%
\end{equation}

Clearly, we have:
\begin{equation}
(x,w)\in U^{\prime}\Longleftrightarrow w\in U. \label{EFThm1aa}%
\end{equation}

Hence:
\begin{equation}
U\neq\varnothing\Longrightarrow U^{\prime}\neq\varnothing\Longrightarrow
\left(  \forall x\in\mathbb{R}^{p},\exists w\in\mathbb{R}^{q}:(x,w)\in
U^{\prime}\right)  . \label{EFThm1b}%
\end{equation}

Now consider \textit{conditions} $(i)$ and $(ii)$ of the theorem. We have the following.\medskip

\noindent$(i)$ \textit{Condition} $(i)$.\medskip

Using (\ref{EFThm1aa}) and Definition \ref{Extended_Polytope_Dfn}, we have:
\begin{equation}
\varphi_{x}(U)=\varphi_{x}(U^{\prime})=\left\{  x\in\mathbb{R}^{p}:\left(
\exists w\in\mathbb{R}^{q}:(x,w)\in U^{\prime}\right)  \right\}
=\mathbb{R}^{p}. \label{EFThm1c}%
\end{equation}
Because $X$ is bounded, we must have:%
\begin{equation}
X\subset\mathbb{R}^{p}. \label{EFThm1d}%
\end{equation}
Combining (\ref{EFThm1c}) and (\ref{EFThm1d}) gives:
\begin{equation}
\varphi_{x}(U)=\mathbb{R}^{p}\neq X. \label{EFThm1e}%
\end{equation}

\noindent$(ii)$ \textit{Condition} $(ii)$.\medskip

(\ref{EFThm1b}) $\Longrightarrow:$
\begin{equation}
\exists x\in\mathbb{R}^{d}\backslash X:(\exists w\in\mathbb{R}^{q}:(x,w)\in
U^{\prime}). \label{EFThm1f}%
\end{equation}

(\ref{EFThm1aa}) and (\ref{EFThm1f}) $\Longrightarrow:$
\begin{equation}
(w\in U)\Longleftrightarrow(x,w)\in U^{\prime}\nLeftrightarrow x\in X.
\label{EFThm1g}%
\end{equation}
Hence, the ``\textit{if and only if}'' condition of Definition
\ref{Extended_Polytope_Dfn3} cannot be satisfied in general. \ \ \medskip
\end{proof}

\begin{remark}
\label{Conforti_Rmk2}Theorem \ref{Independent_Spaces_Thm} is consistent with
Remark \ref{Conforti_Rmk} (p. \pageref{Conforti_Rmk}).
\end{remark}

\begin{corollary}
\label{Ill_Dfn_Coroll_1}Let $X:=$ $\left\{  x\in\mathbb{R}^{p}:Ax\leq
a\right\}  $ and $U:=\left\{  w\in\mathbb{R}^{q}:Hw\leq h\right\}  $ be
(non-empty) polytopes in \textit{independent spaces} (where: $A\in
\mathbb{R}^{m\times p};$ $a\in\mathbb{R}^{m};$ $H\in\mathbb{R}^{n\times q};$
$h\in\mathbb{R}^{n}$)\textit{.} Then, exactly one of the following is true:

\begin{description}
\item $(i)$ There exists no \textit{extended formulation }relationship between
$X$ and $U$ (i.e., there exists no linear map $\pi_{x}:\mathbb{R}%
^{q}\longrightarrow\mathbb{R}^{p}$ such that $\pi_{x}(U)=X,$ and there exists
no linear map $\pi_{w}:\mathbb{R}^{p}\longrightarrow\mathbb{R}^{q}$ such that
$\pi_{w}(X)=U$)$;$

\item $(ii)$ The \textit{extended formulation }relationship between $X$ and
$U$ is \textit{ill-defined} due to Definition \ref{Extended_Polytope_Dfn2}
being inconsistent with Definitions \ref{Extended_Polytope_Dfn} and
\ref{Extended_Polytope_Dfn3}, respectively (i.e., if there exists a linear map
$\pi_{x}:\mathbb{R}^{q}\longrightarrow\mathbb{R}^{p}$ such that $\pi
_{x}(U)=X,$ or there exists a linear map $\pi_{w}:\mathbb{R}^{p}%
\longrightarrow\mathbb{R}^{q}$ such that $\pi_{w}(X)=U,$ or both).\medskip
\end{description}
\end{corollary}

\begin{example}
\label{Ill_Dfn_Coroll_Illustr}Corollary \ref{Ill_Dfn_Coroll_1}.$ii$ can be
illustrated using the polytopes $P$ and $Q$ of Example
\ref{Indep_Spaces_Illustr}. \medskip

\noindent We have:

\begin{description}
\item $(i)$ $\ \varphi_{y}(P)=\mathbb{R}^{2}\neq Q$ and \ $\varphi
_{x}(Q)=\mathbb{R}^{2}\neq P.$ \medskip\newline Hence, according Definition
\ref{Extended_Polytope_Dfn}, there exists no extension relationship between
$P$ and $Q$;

\item $(ii)$ $\ \exists x\notin P:(\exists y\in\mathbb{R}^{2}:(x,y)\in Q),$
which implies: $(x\in P\nLeftrightarrow(\exists$ $y\in\mathbb{R}^{2}:(x,y)\in
Q)).$ Similarly, $\exists y\notin Q:(\exists x\in\mathbb{R}^{2}:(x,y)\in Q),$
which implies: $(y\in Q\nLeftrightarrow(\exists$ $x\in\mathbb{R}^{2}:(x,y)\in
P)).\medskip$\newline Hence, according Definition \ref{Extended_Polytope_Dfn3}%
, there exists no extension relationship between $P$ and $Q$;

\item $(iii)$ $(x,y)\in(P,Q)\Longrightarrow$ $\left[
\begin{array}
[c]{c}%
x_{1}\\
x_{2}%
\end{array}
\right]  =\left[
\begin{array}
[c]{cc}%
1 & 1\\
0 & 0
\end{array}
\right]  \left[
\begin{array}
[c]{c}%
y_{1}\\
y_{2}%
\end{array}
\right]  .$ \newline In other words, $(x,y)\in(P,Q)\Longrightarrow x=Ay$,
where $A=\left[
\begin{array}
[c]{cc}%
1 & 1\\
0 & 0
\end{array}
\right]  $ is the matrix for a linear transformation which maps $P$ and
$Q.\medskip$\newline Hence, according to Definition
\ref{Extended_Polytope_Dfn2}, $Q$ is an \textit{extended formulation} of $P$,
which is in contradiction of $(i)$ and $(ii)$ above.
\end{description}

\noindent$\square$ \ 
\end{example}

The \textit{ill-definition} condition stated in Corollary
\ref{Ill_Dfn_Coroll_1} will be further developed in the remainder of this
section. We start with a notion which essentially generalizes the idea of the
linear map $(\pi)$ in Definition \ref{Extended_Polytope_Dfn2} with respect to
the task of optimizing a linear function over a polyhedral set (since each of
the linear equations that specify $\pi$ must be \textit{valid} for $U$ and
$X$, respectively).

\begin{theorem}
\label{Extended_Formulation_Thm}Any two non-empty polytopes expressed in
\textit{independent spaces} can be respectively \textit{augmented} into being
\textit{extended formulations} of each other. In other words, let $x^{1}%
\in\mathbb{R}^{n_{1}}$ ($n_{1}\in\mathbb{N}_{+}$) and $x^{2}\in\mathbb{R}%
^{n_{2}}$ ($n_{2}\in\mathbb{N}_{+}$) be vectors of variables in two
\textit{independent spaces}. Then, every non-empty polytope in $x^{1}$ can be
\textit{augmented} into an \textit{extended formulation} of every other
non-empty polytope in $x^{2}$, and vice versa.
\end{theorem}

\begin{proof}
\ The proof is essentially by construction. \medskip

\noindent Let $P_{1}$ and $P_{2}$ be polytopes specified as:%
\begin{align*}
&  P_{1}=\{x^{1}\in\mathbb{R}^{n_{1}}:A_{1}x^{1}\leq a_{1}\}\neq
\varnothing\text{ \ }(\text{where }A_{1}\in\mathbb{R}^{p_{1}\times n_{1}%
}\text{, and }a_{1}\in\mathbb{R}^{p_{1}});\\[0.09in]
&  P_{2}=\{x^{2}\in\mathbb{R}^{n_{2}}:A_{2}x^{2}\leq a_{2}\}\neq
\varnothing\text{ \ (where }A_{2}\in\mathbb{R}^{p_{2}\times n_{2}}\text{, and
}a_{2}\in\mathbb{R}^{p_{2}}).
\end{align*}

Clearly, $\forall(x^{1},$ $x^{2})\in P_{1}\times P_{2},$ $\forall
q\in\mathbb{N}_{\mathbb{+}}$, $\forall B_{1}\in\mathbb{R}^{q}{}^{\times n_{1}%
},$ $\forall B_{2}\in\mathbb{R}^{q\times n_{2}},$ there exists $u\in
\mathbb{R}_{\nless}^{q}$ such that the constraints%
\begin{equation}
B_{1}x^{1}+B_{2}x^{2}-u\leq0 \label{EF_Thm(a)}%
\end{equation}
are \textit{valid} for $P_{1}$ and $P_{2}$, respectively (i.e., they are
\textit{redundant} for $P_{1}$ and $P_{2},$ respectively).\medskip

Now, consider :%
\begin{align}
W:=  &  \left\{  (x^{1},x^{2},u)\in\mathbb{R}^{n_{1}}\times\mathbb{R}^{n_{2}%
}\times\mathbb{R}_{\nless}^{q}:\right. \nonumber\\[0.06in]
&  C_{1}A_{1}x^{1}\leq C_{1}a_{1};\text{ }\label{EF_Thm(b)}\\[0.06in]
&  B_{1}x^{2}+B_{2}x^{1}-u\leq0;\label{EF_Thm(c)}\\[0.06in]
&  \left.  C_{2}A_{2}x^{2}\leq C_{2}a_{2}\right\}  \label{EF_Thm(d)}%
\end{align}
(where: $C_{1}\in$ $\mathbb{R}^{p_{1}\times}{}^{p_{1}}$ and $C_{2}$
$\in\mathbb{R}^{p_{2}\times}{}^{p_{2}}$ are diagonal matrices with non-zero
diagonal entries).\medskip

Clearly, $W$ \textit{augments} $P_{1}$ and\textit{\ }$P_{2}$ respectively.
Hence:%
\begin{equation}
W\text{ is equivalent to }P_{1},\text{ and} \label{EF_Thm(g)}%
\end{equation}%
\begin{equation}
W\text{ is equivalent to }P_{2}\text{ .} \label{EF_Thm(h)}%
\end{equation}

Also clearly, we have:%
\begin{equation}
\varphi_{x^{1}}(W)=P_{1}\text{ \ (since }P_{2}\neq\varnothing,\text{ and
((\ref{EF_Thm(c)}) and (\ref{EF_Thm(d)}) are redundant for }P_{1}%
)),\text{\ \ and } \label{EF_Thm(e)}%
\end{equation}%
\begin{equation}
\varphi_{x^{2}}(W)=P_{2}\text{ \ (since }P_{1}\neq\varnothing,\text{ and
((\ref{EF_Thm(b)}) and (\ref{EF_Thm(c)}) are redundant for }P_{2})\text{)}.
\label{EF_Thm(f)}%
\end{equation}

It follows from the combination of (\ref{EF_Thm(g)}) and (\ref{EF_Thm(f)})
that $P_{1}$ is an \textit{extended formulation} of $P_{2}.\medskip$

It follows from the combination of (\ref{EF_Thm(h)}) and (\ref{EF_Thm(e)})
that $P_{2}$ is an \textit{extended formulation} of $P_{1}.$\ \ \medskip\ 
\end{proof}

\begin{corollary}
\label{EF_Corollary}Provided polytopes can be arbitrarily \textit{augmented}
for the purpose of establishing \textit{extended formulation} relationships,
every two (non-empty) polytopes expressed in \textit{independent spaces} are
\textit{extended formulations} of each other. \medskip
\end{corollary}

Theorem \ref{Extended_Formulation_Thm} and Corollary \ref{EF_Corollary} are
illustrated below.

\begin{example}
\label{Extended_Formulation_Example}\ \ \medskip\newline Let
\begin{align*}
&  P_{1}=\{x\in\mathbb{R}_{\nless}^{2}:2x_{1}+x_{2}\leq6\};\\[0.06in]
&  P_{2}=\{w\in\mathbb{R}_{\nless}^{3}:18w_{1}-w_{2}\leq23;\text{ }%
59w_{1}+w_{3}\leq84\}.
\end{align*}
For arbitrary matrices $B_{1},$ $B_{2}$, $C_{1},$ and $C_{2}$ (of appropriate
dimensions, respectively)$;$ say $B_{1}=\left[
\begin{array}
[c]{cc}%
-1 & 2\\
3 & -4
\end{array}
\right]  ,$ $B_{2}=\left[
\begin{array}
[c]{ccc}%
5 & -6 & 7\\
-10 & 9 & -8
\end{array}
\right]  ,$ $C_{1}=\left[  7\right]  ,$ and $C_{2}=\left[
\begin{array}
[c]{cc}%
2 & 0\\
0 & 0.5
\end{array}
\right]  ;$ $P_{1}$ and $P_{2}$ can be \textit{augmented} into
\textit{extended formulations} of each other using $u\in\mathbb{R}_{\nless
}^{2}$ and $W$:
\begin{align*}
W=  &  \left\{  (x,w,u)\in\mathbb{R}_{\nless}^{2+3+2}:\text{ \ }\left[
7\right]  \left[
\begin{array}
[c]{cc}%
2 & 1
\end{array}
\right]  \left[
\begin{array}
[c]{c}%
x_{1}\\
x_{2}%
\end{array}
\right]  \leq42\right.  ;\\
&  \left[
\begin{array}
[c]{cc}%
-1 & 2\\
3 & -4
\end{array}
\right]  \left[
\begin{array}
[c]{c}%
x_{1}\\
x_{2}%
\end{array}
\right]  +\left[
\begin{array}
[c]{ccc}%
5 & -6 & 7\\
-10 & 9 & -8
\end{array}
\right]  \left[
\begin{array}
[c]{c}%
w_{1}\\
w_{2}\\
w_{3}%
\end{array}
\right]  -\left[
\begin{array}
[c]{c}%
u_{1}\\
u_{2}%
\end{array}
\right]  \leq\left[
\begin{array}
[c]{c}%
0\\
0
\end{array}
\right]  ;\\
&  \left.  \left[
\begin{array}
[c]{cc}%
2 & 0\\
0 & 0.5
\end{array}
\right]  \left[
\begin{array}
[c]{ccc}%
18 & -1 & 0\\
59 & 0 & 1
\end{array}
\right]  \left[
\begin{array}
[c]{c}%
w_{1}\\
w_{2}\\
w_{3}%
\end{array}
\right]  \leq\left[
\begin{array}
[c]{c}%
46\\
42
\end{array}
\right]  \right\}  .\text{ \ }%
\end{align*}
$\square$
\end{example}

\begin{remark}
\ \ 

\begin{enumerate}
\item According to Corollary \ref{EF_Corollary}, the notion of EF becomes
degenerate when $G=\mathbf{0}$ in Definitions \ref{Extended_Polytope_Dfn},
\ref{Extended_Polytope_Dfn2}, and \ref{Extended_Polytope_Dfn3}, respectively,
and one tries to apply it by \textit{augmenting} one of the polytopes at hand.

\item Theorem \ref{Extended_Formulation_Thm} and Corollary \ref{EF_Corollary}
are not true for polytopes expressed in \textit{overlapping spaces}, and that
in fact, these two results are in contradiction of Theorem
\ref{Overlapg_Polytopes_Thm}. Hence, whereas one can arbitrarily augment
polytopes in \textit{overlapping spaces} for the purpose of establishing EF
relationships, such an approach is invalid (cannot produce valid results) for
polytopes stated in \textit{independent spaces.\medskip}
\end{enumerate}

Clearly, any notion of ``extension'' which allows for an object to be
extensions of its own extensions cannot be a we\textit{ll-defined} one (i.e.,
must be an \textit{ill-defined} one), unless the objects involved are
indistinguishable from their respective ``extensions.'' For example, clearly,
one cannot reasonably argue that $P_{1}$ and $P_{2}$ in Example
\ref{Extended_Formulation_Example} above are \textit{extended formulations} of
each other in a meaningful sense.
\end{remark}

\section{Redundancy matters when relating polytopes stated in
\textit{independent spaces\label{Redundancy_Matters_Section}}}

The notion of \textit{independent spaces} we have introduced in this paper is
important because, as we have shown, it refines the notion of EFs by
separating the case in which that notion is degenerate (with every polytope
potentially being an EF of every other polytope) from the case where the
notion of EF is \textit{well-defined}/meaningful. It separates the case in
which the addition of redundant constraints and variables (for the purpose of
establishing EF relations) matters (i.e., makes a difference to the outcome of
analysis) from the case in which the addition of redundant constraints and
variables does not matter. \medskip

Two key results of section \ref{Ill_Definition_Section}\ of this papers are that:

\begin{enumerate}
\item If $U$ in $\mathbb{R}^{p}$ and $V$ in $\mathbb{R}^{q}$ are in
\textit{overlapping spaces}, then an \textit{augmentation} of $V$ is an EF of
\textit{U} if and only if $V$ is an EF of $U$; (This is the case where the
addition of redundant constraints does not matter, and is stated in Theorem
\ref{Overlapg_Polytopes_Thm} on page \pageref{Overlapg_Polytopes_Thm} of this paper);

\item But ($1$) is not true if $U$ and $V$ are in \textit{independent spaces}.
As we have shown in section \ref{Ill_Dfn_Condition_Subsection}, if polytopes
can always be arbitrarily augmented for the purpose of establishing EF
relations, then any two polytopes that are in \textit{independent spaces} are
EF's of each other. (This is the case where adding redundant constraints and
variables does matter, and is stated in Theorem \ref{Extended_Formulation_Thm}
on page \pageref{Extended_Formulation_Thm} of this paper).
\end{enumerate}

Because of ($2$) above, the addition of redundant constraints and variables
for the purpose of establishing EF relations can lead to ambiguities when
applied to polytopes in \textit{independent spaces}. This ambiguity would stem
from the fact that one would reach contradicting conclusions depending on what
we do with the redundant constraints and variables which are introduced in
order to make the models \textit{overlap}. To clarify this: Assume V is
augmented with the variables of $U$ plus the constraints for the linear
transformation that establishes the $1-1$ correspondence between $U$ and $V$;
call this augmented-$V$, $V^{\prime}$. In EF work it is commonly
assumed/suggested that redundant constraints and variables of a model can be
removed from it without any loss of generality. In the case of $V^{\prime}$
this would lead to contradictory conclusions with respect to the question of
whether or not $V^{\prime}$ is an EF of $U$: If the added redundant
constraints and variables are kept, the answer would be ``yes''; If they are
removed, $V^{\prime}$ would ``revert'' back to $V$, so that the answer would
be ``no.'' This is the inconsistency which is pointed out in Theorem
\ref{Independent_Spaces_Thm} (p. \pageref{Independent_Spaces_Thm}) and
Corollary \ref{Ill_Dfn_Coroll_1} (p. \pageref{Ill_Dfn_Coroll_1}), and
illustrated in Example \ref{Ill_Dfn_Coroll_Illustr} (p.
\pageref{Ill_Dfn_Coroll_Illustr}).\medskip

A specific implication of ($2$) above is that, provided polytopes can be
arbitrarily augmented for the purpose of EF's, every conceivable polytope that
is non-empty and does not require the \textit{traditional }$x$%
\textit{-variables} (i.e., the city-to-city $x_{i,j}$ variables of the
\textit{Standard TSP Polytope}; see Definitions \ref{traditional x-variables}
and \ref{Standard_TSP_Polytope} (p. \pageref{Standard_TSP_Polytope})) in its
description is an EF of the \textit{Standard TSP Polytope}, and vice versa.
Clearly, this can only be in a degenerate/non-meaningful sense from which no
valid inferences can be made.

\subsection{The case of the Minimum Spanning Tree
Problem\label{Min_Span_Tree_SubSection}}

A case in point for the discussions in the introduction to section
\ref{Redundancy_Matters_Section} above, is that of the Minimum Spanning Tree
Problem (MSTP). Without the refinement that our notion of \textit{independent
spaces} contributes to the notion of EFs, this case (of the MSTP) would mean
that it is possible to \textit{extend} an exponential-sized model into a
polynomial-sized one by (simply) \textit{augmenting}\ the exponential model,
which is a clearly-unreasonable/out-of-the-question proposition. This
proposition would be arrived at as follows. Assume (as is normally done in EF
work) that the addition of redundant constraints and variables does not matter
as far EF theory is concerned. Since the constraints of Edmonds' model
(Edmonds (1970)) are redundant for the model of Martin (1991), one could
\textit{augment} Martin's formulation with these constraints. The resulting
model would still be considered a polynomial-sized one. But note that this
particular \textit{augmentation }of Martin's model would also be an
\textit{augmentation} of Edmonds' model. Hence, the conclusion would be that
Edmonds' exponential-sized model has been \textit{augmented} into a
polynomial-sized one, which is an impossibility, since one cannot reduce the
number of facets of a given polytope by \textit{augmenting} that polytope. The
refinement brought by our notion of \textit{independent spaces} explains this
paradox in the case of MSTP, as shown below.

\begin{example}
\label{EF_MST_Example}We show that Martin's polynomial-sized LP model of the
MSTP is not an EF (in a non-degenerate, meaningful sense) of Edmonds's
exponential LP model of the MSTP, by showing that Martin's model can be stated
in \textit{independent space} relative to Edmonds' model.
\end{example}

\begin{itemize}
\item \textbf{Using the notation in Martin(1991), i.e.:}

\begin{itemize}
\item $N:=\{1,\ldots,n\}$ \ \ (Set of vertices);

\item $E:$ \ Set of edges;

\item $\forall S\subseteq N,$ $\gamma(S):$ Set of edges with both ends in $S $.\medskip
\end{itemize}

\item \textbf{Exponential-sized/``sub-tour elimination'' LP formulation
(Edmonds (1970))}
\end{itemize}

$(P)$:\medskip\newline
\begin{tabular}
[c]{l}%
\ \ \
\end{tabular}
$\left|
\begin{tabular}
[c]{ll}%
$\text{Minimize:}$ & $\sum\limits_{e\in E}c_{e}x_{e}$\\
& \\
$\text{Subject To:}$ & $\sum\limits_{e\in E}x_{e}=n-1;$\\
& \\
& $\sum\limits_{e\in\gamma(S)}x_{e}\leq\left|  S\right|  -1;$ \ $\ S\subset
E$\ $;$\\
& \\
& $x_{e}\geq0$ \ for all $e\in E.$%
\end{tabular}
\text{ \ }\right.  $\medskip

\begin{itemize}
\item \textbf{Polynomial-sized LP reformulation (Martin (1991))}
\end{itemize}

$(Q)$:\medskip\newline
\begin{tabular}
[c]{l}%
\ \ \
\end{tabular}
$\left|
\begin{tabular}
[c]{ll}%
$\text{Minimize:}$ & $\sum\limits_{e\in E}c_{e}x_{e}$\\
& \\
$\text{Subject To:}$ & $\sum\limits_{e\in E}x_{e}=n-1;$\\
& \\
& $z_{k,i,j}+z_{k,j,i}=x_{e};$ \ \ $k=1,\ldots,n;$ \ $e\in\gamma(\{i,j\}); $\\
& \\
& $\sum\limits_{s>i}z_{k,i,s}+\sum\limits_{h<i}z_{k,i,h}\leq1;$
$\ \ k=1,\ldots,n;$ $\ \ i\neq k;$\\
& \\
& $\sum\limits_{s>k}z_{k,k,s}+\sum\limits_{h<k}z_{k,k,h}\leq0;$ \ $k=1,\ldots
,n;$\\
& \\
& $x_{e}\geq0$ \ for all $e\in E$; \ \ $z_{k,i,j}\geq0$ \ for all $k,$ $i,$
$j.$%
\end{tabular}
\text{ \ }\right.  \medskip$\ 

\begin{itemize}
\item \textbf{Re-statement of Martin's LP model (Diaby and Karwan (2013);
Regan (2013))\medskip}
\end{itemize}

For each $e\in E:\medskip$

\qquad- Denote the ends of $e$ as $i_{e}$ and $j_{e},$respectively;\medskip

\qquad- Fix an arbitrary node, $r_{e}$, which is not incident on $e$ (i.e.,
$r_{e}$ is such that it is not an end of $e$). \medskip

Then, one can verify that $Q$ is equivalent to:\medskip\newline $(Q\prime
)$:\medskip\newline
\begin{tabular}
[c]{l}%
\ \ \
\end{tabular}
$\left|
\begin{tabular}
[c]{ll}%
$\text{Minimize:}$ & $\sum\limits_{e\in E}c_{e}z_{r_{e},i_{e},j_{e}}%
+\sum\limits_{e\in E}c_{e}z_{r_{e},j_{e},i_{e}}$\\
& \\
$\text{Subject To:}$ & $\sum\limits_{e\in E}z_{r_{e},i_{e},j_{e}}%
+\sum\limits_{e\in E}z_{r_{e},j_{e},i_{e}}=n-1;$\\
& \\
& $z_{k,i_{e},j_{e}}+z_{k,j_{e},i_{e}}=z_{r_{e},i_{e},j_{e}}+z_{r_{e}%
,j_{e},i_{e}}; $ \ \ $k=1,\ldots,n;$\ $\ \ e\in E;$\\
& \\
& $\sum\limits_{s>i}z_{k,i,s}+\sum\limits_{h<i}z_{k,i,h}\leq1;$ \ \ $i,$
$k=1,\ldots,n:i\neq k;$\\
& \\
& $\sum\limits_{s>k}z_{k,k,s}+\sum\limits_{h<k}z_{k,k,h}\leq0;$
\ \ $k=1,\ldots,n;$\\
& \\
& $z_{k,i,j}\geq0$ \ for all $k,$ $i,$ $j.$%
\end{tabular}
\text{ \ }\right.  $

\noindent$\square\medskip$

\begin{remark}
\label{EF_MSTP_Rmk}We argue that the reason that EF modeling ``barriers'' do
not apply in the case of the MSTP is due to the fact\ that Martin's
formulation of the MSTP ($Q$ in Example \ref{EF_MST_Example} above) is not an
EF of Edmonds' model ($P$ in Example \ref{EF_MST_Example} above) in a
non-degenerate/meaningful, \textit{well-defined} sense. The reason for this in
turn, is that Martin's formulation \textit{can be} stated in
\textit{independent space} relative to Edmonds' model, due to Martin's
formulation having \textit{column-redundancy} when it includes the
\textit{class of variables} of Edmonds' model (see Definition
\ref{Column_Redundancy_Dfn}), as shown in Example \ref{EF_MST_Example}.
\end{remark}

\subsection{Alternate/Auxiliary Models}

In this section, we provide some insights into the meaning of the existence of
an affine map establishing a one-to-one correspondence between polytopes that
are stated in \textit{independent spaces} as brought to our attention in
private discussions by Yannakakis (2013). The linear map stipulated in
Definition \ref{Extended_Polytope_Dfn2} is a special case of the affine map.
Referring back to Definitions \ref{Extended_Polytope_Dfn}%
-\ref{Extended_Polytope_Dfn3}, assume $G=\mathbf{0}$ in the expression of $U$.
We will show in this section, that in that case, provided the matrix of the
affine transformation does not have any strictly-negative entry, $U$ is simply
an alternate model (a ``reformulation'') of $P$ which can be used, in an
``auxiliary'' way, in order to solve the optimization problem over $P$ without
any reference to/knowledge of the $\mathcal{H}$-description of $P$ (see Remark
\ref{Alternate_TSP_Polytope_Rmk} (p. \pageref{Alternate_TSP_Polytope_Rmk}) of
this paper).

\begin{remark}
\label{EF_Insight_Rmk1} \ 

\begin{itemize}
\item Referring back to Example \ref{Polytope_Augmentation_Example} (p.
\pageref{Polytope_Augmentation_Example}), assume that the non-negativity
requirements for $x$ and $y$ are included in the constraints of $X$ and $Y$,
respectively, and that $L$ has the form:%
\[
L=\{(x,y)\in\mathbb{R}_{\nless}^{p+q}:x-Cy=b\}\text{ \ \ (where }%
C\in\mathbb{R}_{\nless}^{p\times q}\text{, and }b\in\mathbb{R}^{p}).
\]

\item Consider the optimization problem:\medskip
\end{itemize}

\textit{Problem LP}$_{1}$:\medskip%

\begin{tabular}
[c]{l}%
\ \ \
\end{tabular}
$\left|
\begin{tabular}
[c]{ll}%
$\text{Minimize:}$ & $\alpha^{T}x$\\
& \\
$\text{Subject To:}$ & $(x,y)\in L;$ $\ y\in Y$\\
& \\
\multicolumn{2}{l}{(where $\alpha\in\mathbb{R}^{p}).$}%
\end{tabular}
\text{ \ }\right.  \medskip$

\begin{itemize}
\item \textit{Problem LP}$_{1}$ is equivalent to the smaller linear program:\medskip
\end{itemize}

\textit{Problem LP}$_{2}$:\medskip\medskip%

\begin{tabular}
[c]{l}%
\ \ \
\end{tabular}
$\left|
\begin{tabular}
[c]{ll}%
$\text{Minimize:}$ & $\left(  \alpha^{T}C\right)  y+\alpha^{T}b$\\
& \\
$\text{Subject To:}$ & $y\in Y$\\
& \\
\multicolumn{2}{l}{(where $\alpha\in\mathbb{R}^{p}).$}%
\end{tabular}
\text{ \ }\right.  \medskip\medskip$

Hence, if $L$ is the graph of a one-to-one correspondence between the points
of $X$ and the points of $Y$ (see Beachy and Blair (2006, pp. 47-59)), then,
the optimization of any linear function of $x$ over $X$ can be done by first
using \textit{Problem LP}$_{\mathit{2}}$ in order to get an optimal $y,$ and
then using Graph $L$ to ``retrieve'' the corresponding $x$. Note that the
second term of the objective function of \textit{Problem LP}$_{\mathit{2}}$
can be ignored in the optimization process of \textit{Problem LP}%
$_{\mathit{2}},$ since that term is a constant.\medskip

Hence, if $L$ is derived from knowledge of the $\mathcal{V}$-representation of
$X$ only (as is the case for the TSP\ LP\ models of Diaby (2007) and Diaby and
Karwan (2012) relative to the \textit{Standard TSP Polytope}), then this would
mean that the $\mathcal{H}$-representation of $X$ is not involved in the
``two-step'' solution process (of using \textit{Problem LP}$_{\mathit{2}}$ and
then Graph $L$), but rather, that only the $\mathcal{V}$-representation of $X$
is involved (see Remark \ref{Alternate_TSP_Polytope_Rmk} (p.
\pageref{Alternate_TSP_Polytope_Rmk}) of this paper). The case of the MSTP can
be used to illustrate this point also. Referring back to Example
\ref{EF_MST_Example}, clearly, it is not possible to get the $\mathcal{H}%
$-\textit{description} $P$ by simple mathematical ``manipulations'' of the
$\mathcal{H}$-\textit{description} $Q\prime$. One could derive $Q$ from
$Q\prime$ \textbf{only} by establishing correspondences which are based on the
knowledge of the $\mathcal{V}$\textit{-descriptions} of $P$ and $Q\prime$. In
other words, although solutions of $P$ can be ``retrieved'' from those of
$Q\prime$, that ``retrieval'' is based on knowledge of the $\mathcal{V}%
$-\textit{descriptions }involved and does not, therefore, impy any meaningful
\textit{extension} relationships between the $\mathcal{H}$%
-descriptions\textit{.}\medskip

Hence, in general, when $G=\mathbf{0}$ in Definition
\ref{Extended_Polytope_Dfn2}$,$ the condition stipulated in that definition
cannot imply (or lead to) \textit{extended formulation} relations which are
meaningful in relating the minimal $\mathcal{H}$-representations of the
polytopes involved (see Example \ref{Indep_Spaces_Illustr} (p.
\pageref{Indep_Spaces_Illustr})\ also). $\ \ \square$
\end{remark}

Direct corollaries of the developments above in this section and in section
\ref{Ill_Definition_Section} are the following.

\begin{corollary}
\label{EF_Insight_Coroll}Let $P$ and $Q$ be (non-empty) polytopes stated in
\textit{overlapping spaces.} Assume w.l.o.g. that the set of the descriptive
variables of $P$ is embedded in the set of the descriptive variables of $Q$.
If $Q$ \textbf{can be} expressed in \textit{independent space} relative to $P
$ (i.e., if all of the constraints involving the variables of $P$ can be
dropped from $Q$ after the variables of $P$ are substituted out of the
objective function of the optimization problem over $Q$), then $Q$ is not (and
cannot be \textit{augmented} into) a \textit{well-defined} (non-degenerate,
non-ambiguous, meaningful) EF of $P$.
\end{corollary}

\begin{corollary}
\label{EF_Insight_Coroll2}\textit{Extended Formulations} developments relating
problem sizes (such as Yannakakis (1991), and Fiorini et al. (2011; 2012), in
particular) are valid/applicable only when the projections involved are
irredundant-component projections.\medskip
\end{corollary}

\subsection{Application to the Fiorini \textit{et al}. (2011; 2012)
``barriers'' \label{Barriers_Subsection}}

Having addressed the MSTP in section \ref{Min_Span_Tree_SubSection}, we will
return our focus to the TSP in this section. We will illustrate Corollaries
\ref{EF_Insight_Coroll} and \ref{EF_Insight_Coroll2} using the developments in
Fiorini et al. (2011; 2012), by showing that the mathematics in those papers
actually ``breaks down'' as one tries to apply their developments when the
polytopes involved are stated in \textit{independent spaces} (i.e., when
$G=\mathbf{0}$ in Definitions \ref{Extended_Polytope_Dfn},
\ref{Extended_Polytope_Dfn2}, \ref{Extended_Polytope_Dfn3}, respectively). As
we indicated in the Introduction section (section \ref{Introduction_Section}),
in this paper, we are not concerned with the issue of
correctness/incorrectness of any particular LP model that may have been
proposed for NP-Complete problems. Rather, our aim is to show that the
resolution of that issue (of correctness/incorrectness) can be beyond the
scope of EF work under some conditions, such as is the case for the LP models
of Diaby (2007), and Diaby and Karwan (2012), for example. In order to
simplify the discussion, we will focus on the \textit{Standard TSP Polytope},
and use the \textit{Alternate TSP Polytope} discussed in this paper (see
\ref{Alternate_Polytope_Dfn}, p. \pageref{Alternate_Polytope_Dfn}), as well as
the TSP LP models of Diaby (2007) and Diaby and Karwan (2012), respectively,
as illustrations. \medskip

Fiorini \textit{et al}. (2012) is a re-organized and extended version of
Fiorini \textit{et al}. (2011). The key extension is the addition of another
alternate defnition of \textit{extended formulation} (page 96 of Fiorini
\textit{et al}. (2012)) which is recalled in this paper as Definition
\ref{Extended_Polytope_Dfn3}. This new alternate definition is then used to
re-arrange ``section 5'' of Fiorini \textit{et al}. (2011) into ``section 2''
and ``section 3'' of Fiorini \textit{et al}. (2012). Hence, the developments
in ``section 5'' of Fiorini \textit{et al}. (2011) which depended on ``Theorem
4'' of that paper, are ``stand-alones'' (as ``section 3'') in Fiorini
\textit{et al}. (2012), and ``Theorem 4'' in Fiorini \textit{et al}. (2011) is
relabeled as ``Theorem 13'' in Fiorini \textit{et al}. (2012).\medskip

Our discussion of specifics why neither of the two papers are applicable when
relating polytopes in \textit{independent spaces} will be based on the
non-validity of the proofs of ``Theorem 4'' of Fiorini \textit{et al}. (2011)
(which is ``Theorem 13'' of Fiorini \textit{et al}. (2012), as indicated
above), and of ``Theorem 3'' of Fiorini \textit{et al}. (2012) (which is in
``section 3'' of that paper) when $G=\mathbf{0}$ in Definitions
\ref{Extended_Polytope_Dfn}, \ref{Extended_Polytope_Dfn2}, and
\ref{Extended_Polytope_Dfn3}, respectively.

\begin{theorem}
\label{Fiorini_et_al_2011_Thm}Let $W$ $\subset\mathbb{R}^{\xi},$ be the
polytope involved in an arbitrary abstraction of TSP tours. Assume $W$ and the
\textit{Standard TSP Polytope} are expressed in \textit{independent spaces
}(such as is the case for the \textit{Alternate TSP Polytope}, or the
polytopes associated with the LP models of the TSP proposed in Diaby (2007)
and in Diaby and Karwan (2012), respectively)\textit{.} Then, the developments
in Fiorini \textit{et al}. (2011) are not valid (and therefore, not
applicable) for relating the size of $W$ to the size of the \textit{Standard
TSP Polytope}.
\end{theorem}

\begin{proof}
Using the terminology and notation of Fiorini \textit{et al.} (2011), the main
results of section 2 of Fiorini \textit{et al.} (2011) are developed in terms
of $Q:=\{(x,y)\in\mathbb{R}^{d+k}$ $\left|  \text{ }Ex+Fy=g,\text{ }y\in
C\right.  \}$ and $P:=\{x\in\mathbb{R}^{d}$ $\left|  \text{ }Ax\leq b\right.
\}.$ \medskip

Note that letting $Q$ (in Fiorini \textit{et al. }(2011)) stand for $W,$ and
$P$ (in Fiorini \textit{et al. }(2011)) stand for the \textit{Standard TSP
Polytope} respectively, $E$ would be equal to $\mathbf{0}$ in the expression
of $Q$. Hence, firstly, assume $E=\mathbf{0}$ in the expression of $Q$ (i.e.,
$Q:=\{(x,y)\in\mathbb{R}^{d+k}$ $\left|  \text{ }\mathbf{0}x+Fy=g,\text{ }y\in
C\right.  \}).$ Then, secondly, consider Theorem 4 of Fiorini \textit{et al.}
(2011) (which is pivotal in that work). We have the following:

\begin{description}
\item $(i)$ If $A\neq\mathbf{0}$ in the expression of $P\mathbf{,}$ then the
proof of the theorem is invalid since that proof requires setting ``$E:=A $''
(see Fiorini \textit{et al.} (2011, p. 7));

\item $(ii)$ If $A=\mathbf{0,}$ then $P:=\{x\in\mathbb{R}^{d}$ $\left|  \text{
}\mathbf{0}x\leq b\right.  \}.$ This implies that either $P=\mathbb{R}^{d}$
(if $b\geq\mathbf{0}$) or $P=\varnothing$ (if $b\ngeq\mathbf{0}$). Hence, $P$
would be either unbounded or empty. Hence, there could not exist a polytope,
$Conv(V),$ such that $P=Conv(V)$ (see Bazaraa \textit{et al}. (2006, pp.
39-49), or Fiorini \textit{et al}. (2011, 16-17), among others). Hence, the
conditions in the statement of Theorem 4 of Fiorini \textit{et al.} (2011)
would be \textit{ill-defined/}impossible.
\end{description}

\noindent Hence, the developments in Fiorini \textit{et al.} (2011) are not
applicable for $W$. \ \ \ \medskip
\end{proof}

\begin{theorem}
\label{Fiorini_et_al_2012_Thm}Let $W$ $\subset\mathbb{R}^{\xi},$ be the
polytope involved in an arbitrary abstraction of TSP tours. Assume $W$ and the
\textit{Standard TSP Polytope} are expressed in \textit{independent spaces}
(such as is the case for the \textit{Alternate TSP Polytope}, or the polytopes
associated with the LP models of the TSP proposed in Diaby (2007) and in Diaby
and Karwan (2012), respectively). Then, the developments in Fiorini \textit{et
al}. (2012) are not valid (and therefore, not applicable) for relating the
size of $W$ to the size of the \textit{Standard TSP Polytope}.
\end{theorem}

\begin{proof}
First, note that ``Theorem 13'' of Fiorini \textit{et al}. (2012, p. 101) is
the same as ``Theorem 4.'' Hence, the proof of Theorem
\ref{Fiorini_et_al_2011_Thm} above is applicable to ``Theorem 13'' of Fiorini
\textit{et al}. (2012). Hence, the developments in Fiorini \textit{et al}.
(2012) that hinge on this result (namely, from ``section 4'' of the paper,
onward) are not applicable to $W$.

Now consider ``Theorem 3'' of Fiorini \textit{et al}. (2012) (section 3, page
99). The proof of this theorem hinges on the statement that (using the
terminology and notation of Fiorini \textit{et al.} (2012)):%
\begin{equation}
Ax\leq b\Longleftrightarrow\exists y:E^{\leq}x+F^{\leq}y\leq g^{\leq},\text{
}E^{=}x+F^{=}y\leq g^{=}. \label{F2012(a)}%
\end{equation}

Now, observe that if $x$ and $y$ do not \textit{overlap},
\[
(\exists y:\mathbf{0}\cdot x+F^{\leq}y\leq g^{\leq},\text{ }\mathbf{0}\cdot
x+F^{=}y\leq g^{=})\text{ cannot imply (}Ax\leq b)\text{ in general.}%
\]

Hence, provided $x$ and $y$ do not \textit{overlap} (i.e., provided $x$ and
$y$ are in \textit{independent spaces}), the ``if and only if'' stipulation of
(\ref{F2012(a)}) cannot be satisfied in general. Hence, Theorem 3 of Fiorini
\textit{et al}. (2012) is not applicable for $W$. \ \ \medskip
\end{proof}

\section{Conclusions\label{Conclusions_Section}}

The developments above in this paper are formalizations of the argument that
the possibility of inferring solutions obtained from any given
\textit{correct} abstraction of a given optimization problem from those of any
other \textit{correct} abstraction of that (same) optimization problem is a
logical necessity, and cannot systematically imply any we\textit{ll-defined}
(non-degenerate, meaningful) \textit{extension} relationships between the
resulting models.\medskip

We would argue that for COP's in general, the paradox in existing
\textit{extended formulations} theory whereby one can \textit{extend} an
exponential-sized model, by \textit{augmenting} it, into a polynomial-sized
model (such as in the case of the \textit{spanning tree polytope}), is due to
the pertinent \textit{extension} relationship being a \textit{ill-defined}
(degenerate, non-meaningful) one\textit{\ }only\textit{. }Indeed, if a given
``extension'' of a model has \textit{row-} and/or \textit{column-redundancies}%
, then there may exist a description of the ``extension'' in question in which
``$G=\mathbf{0}$'' (where $G$ is the matrix of coefficients for the variables
of the ``original'' model, as in Definitions \ref{Extended_Polytope_Dfn},
\ref{Extended_Polytope_Dfn2}, and \ref{Extended_Polytope_Dfn3}). It would be
possible in that case to \textit{substitute }all of the variables of the
``original'' model \textit{out} of the ``extension'' at hand. Hence, in that
case, the ``extension'' at hand would simply be an alternate abstraction of
the problem at hand, stated in \textit{independent space }from the
``original'' model. Hence, fundamentally in that case, there would exist no
\textit{well-defined} (non-degenerate, meaningful) \textit{extension}
relationship between the ``extension'' and the ``original'' model, but only
\textit{ill-defined} ones (including the case of the ``original'' model being
an ``extension'' of the ``extension''), as has been discussed in this paper.\medskip

Hence, overall, we believe that the \textit{ill-definition} condition we have
shown in this paper, with its consequence whereby EF theory is not applicable
when relating models expressed in \textit{independent spaces}, constitutes a
useful step towards a more complete definition of the scope of applicability
for EF's. \ \pagebreak

\end{document}